\documentclass{article}
\usepackage[ruled]{algorithm2e} 
\usepackage{booktabs} 
\usepackage[square,numbers]{natbib}
\usepackage{authblk}
\usepackage{fullpage}
\usepackage{amsmath, amssymb, amsthm}
\usepackage{tikz}
\usetikzlibrary{shapes,arrows,fit,calc,positioning}
\usetikzlibrary{arrows}
\usetikzlibrary{decorations.markings}
\usepackage{enumerate}
\usepackage{hyperref}
\usepackage{graphicx}
\usepackage{caption}
\usepackage{hyperref}
\usepackage{subcaption}
\usepackage{mathtools}
\usepackage{comment}
\usepackage[mathscr]{eucal}
\SetAlFnt{\small}
\SetAlCapFnt{\small}
\SetAlCapNameFnt{\small}
\SetAlCapHSkip{0pt}
\IncMargin{-\parindent}
\usepackage{thm-restate}
\usepackage{subcaption}
\usepackage{array}
\usepackage{float}
\usepackage{amsmath}
\usepackage{comment}
\usepackage[short]{optidef}
\usepackage{xcolor}
\usepackage{hyperref}

\usepackage{float}

\usepackage[noabbrev,capitalize]{cleveref} 
\usepackage{thmtools}
\usepackage{thm-restate}
\newtheorem{example}{Example}

\newtheorem{theorem}{Theorem}[section]
\newtheorem{lemma}[theorem]{Lemma}
\newtheorem{proposition}[theorem]{Proposition}
\newtheorem{corollary}[theorem]{Corollary}

\newtheorem{definition}[theorem]{Definition}
\newtheorem{observation}[theorem]{Observation}

\newcommand{\voters}{\mathcal{V}}
\newcommand{\election}{{\xi}}
\newcommand{\candids}{\mathcal{C}}

\newcommand{\voter}{v}
\newcommand{\metric}{\mathcal{M}}

\newcommand{\done}{{D}}
\newcommand{\dtwo}{{\bar D}}
\newcommand{\shL}{{ \small \texttt {\#}} \ell}
\newcommand{\shR}{{ \small \texttt {\#}} r}
\newcommand{\shX}{{ \small \texttt {\#}} a}

\newcommand{\sob}{\sum_{\voter \in \voters_\election: x_v \in \regB}}

\newcommand{\cost}{\mathbf {sc}}
\newcommand{\regA}{\textsf{A}}
\newcommand{\regB}{\textsf{B}}
\newcommand{\regC}{\textsf{C}}
\newcommand{\regD}{\textsf{D}}
\newcommand{\canL}{\ell}
\newcommand{\canR}{r}

\newcommand{\probL}{\mathtt P_{\canL}}
\newcommand{\probR}{\mathtt P_{\canR}}

\begin{document}

\title{On the Distortion Value of  Elections with Abstention}

\author[1,2]{Mohammad Ghodsi}
\author[1]{Mohamad Latifian}
\author[2]{Masoud Seddighin}
\affil[1]{Sharif University of Technology}
\affil[2]{Institute for Research in Fundamental Sciences (IPM)}

\maketitle              
%


\maketitle

\begin{abstract}
In Spatial Voting Theory, distortion is a measure of how good the winner is. It has been proved that no deterministic voting mechanism can guarantee a distortion better than $3$, even for simple metrics such as a line. In this study, we wish to answer the following question: how does the distortion value change if we allow less motivated agents to abstain from the election?

We consider an election with two candidates and suggest an abstention model, which is a general form of the abstention model proposed by Kirchg{\"a}ssner  \cite{kirchgassner2003abstention}. Our results  characterize the distortion value and provide a rather complete picture of the model.  \footnote{A preliminary version of this paper is accepted in AAAI 2019.}
\end{abstract}

\section{Introduction}
\label{intro} 
The goal in Social Choice Theory is to design mechanisms that aggregate agents' preferences into a collective decision. 
Voting is a well-studied method for aggregating preferences with many applications in artificial intelligence and multi-agent systems. 
Roughly, a voting mechanism takes the preferences of the agents over a set of alternatives and selects one of them as  winner. 
\newline One fruitful approach to estimate the quality of a voting mechanism is to use the utilitarian view which assumes that each agent has cost over the alternatives \cite{procaccia2006distortion,caragiannis2011voting,boutilier2015optimal,benade2019low,goel2018relating,benade2017preference}. For example,  spatial models locate the voters and the alternatives in a finite metric space $\metric$, and the cost of voter $\voter_i$ for Alternative $x$  equals to their distance \cite{anshelevich2015approximating,anshelevich2016ordinal,anshelevich2017randomized,goel2017metric,black1948rationale,barbera1993generalized,merrill1999unified}.  
Considering these costs, the optimal candidate is defined to be the candidate that minimizes the social cost (the total cost of the voters). Ideally, we would like the optimal candidate to be the winner; however, since voting mechanisms only take the ordinal preferences of voters as input, it is reasonable to expect that the winner is not always optimal. The question then arises: how good is the winner, i.e., what is the worst-case ratio of the social cost of the winner to the social cost of the optimal candidate? This ratio is called the distortion value of a voting mechanism. It is known that no deterministic voting mechanism can guarantee a distortion better than $3$, even for simple metrics such as a line~\cite{anshelevich2015approximating}. To see this, consider the example shown in Figure \ref{3app}. In this example, candidate $\canL$ is the optimal candidate, and under the plurality voting rule
\footnote{For two candidates, all the well-known  deterministic voting mechanisms (e.g. Borda,  $k$-approval, Copeland, etc) turn into plurality.}
candidate $\canR$ is the winner. Thus, the distortion value is
$$ \frac{0.51 (0.5-\varepsilon) + 0.49 \cdot 1}{0.51 (0.5 + \varepsilon)} \simeq 3.$$ 

\begin{figure}[h]
	\centerline{\includegraphics{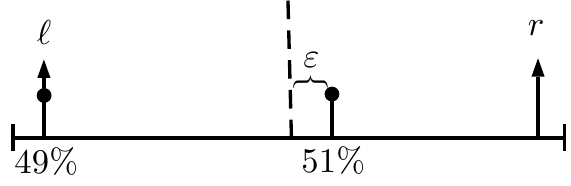}}
	\caption{An example with distortion value close to 3. In this example, $49\%$ of the voters are located at point $0$ and $51\%$ of the voters are located at point $0.5 + \epsilon$. In addition, candidates $\mathsf{L}$ and $\mathsf{R}$ are respectively located at points $0$ and $1$.}
	\label{3app}
\end{figure}

However, the example of Figure \ref{3app} seems unrealistic in some ways. Although the voters located near the point $0.5$ are closer to $\canR$, they have a very low incentive to vote for $\canR$, since their costs for both candidates are almost equal. On the other hand, agents located at $0$ have a strong incentive to vote for $\canL$. Indeed, if voters are allowed to abstain, which is a  natural assumption in many real-world elections, we expect $\canL$ to be the winner rather than $\canR$. 
In this study, our goal is to tackle this problem:
 \begin{quote}
	How does the distortion value change, if we allow less motivated agents to abstain?
\end{quote}

\subsection{ Abstention} \label{abstention}
Scientists have long studied the factors affecting participation in an election. For example, Wolfinger and Rosenstone \cite{feddersen1999abstention} argue that more educated voters participate with a higher probability, or Lijphart \cite{lijphart1997unequal} discusses that the voters on the left side of the political spectrum participate less frequently. Similarly, the decision to vote may rely on variables such as income level or the sense of civic duty \cite{wolfinger1980votes}. 

Traditionally, both game-theoretic and decision-theoretic models of turnout have been proposed. At the heart of most of these models lies the assumption that there are costs for voting.\footnote{There are other decision theoretic explanations of abstention that do not rely on costs, e.g., see \cite{ghirardato1997indecision}.}  These costs include the costs of collecting and processing information, waiting in the queue and voting itself. Presumably, if a voter decides to abstain, she does not have to pay these costs. Therefore, a rational voter must receive utility from voting. There is evidence suggesting that voters behave strategically when deciding to vote and take the costs and benefits into account. For example, Riker and Ordeshook \cite{riker1968theory} show that the turnout is inversely related to voting costs. 

Apart from social-psychological traits, other studies suggest that voters' abstention may stem from their ideological distances from the candidates. The work of Downs \cite{downs1957economic} initiated this line of research. He argues that in a two-candidate election under the majority rule, the choice between voting and abstaining is related to the voter's comparative evaluation of the candidates. Riker and Ordeshook \cite{riker1968theory} later improve this model by reformulating the original equation to incorporate other social and psychological factors.

Many empirical studies in spatial theory of abstention  suggest that the voters are more likely to abstain when they feel indifferent toward the candidates or alienated from them \cite{kirchgassner2003abstention}. The models introduced by Downs \cite{downs1957economic} and Riker and Ordeshook \cite{riker1968theory} are only capable of explaining  the indifference-based abstention which
occurs when the difference between the costs of candidates for a voter is too small to justify voting costs. On the other hand, these models cannot justify alienation-based abstention, which occurs when a voter is too distant from the alternatives to justify voting costs. To alleviate this, some studies argue that the relative ideological distance plays a more critical role than the absolute distance \cite{kirchgassner2003abstention,geys2006rational}.    Our model of abstention in this paper a generalization of the model introduced by Kirchg{\"a}ssner \cite{kirchgassner2003abstention} which incorporates the relative distances.

\subsection{Our Work}
In this paper, we consider the effect of abstention on the distortion value.  
In our study, there are two candidates, and the voters decide whether to vote or abstain based on a comparison between the cost (i.e., distance) of their preferred alternative and the cost of the other alternative.  We define the concepts of \emph{expected winner} and \emph{expected distortion} to evaluate the distortion of an election in our model.  Our results  characterize the distortion value and provide a complete picture of the model. For the special case that our abstention model conforms exactly to that of Kirchg{\"a}ssner \cite{kirchgassner2003abstention}, we show that the distortion of the expected winner is upper bounded by $1.522$. 

We also give an almost tight upper bound on the expected distortion value of large elections. We show that for any $\alpha>0$ and a large enough election (in term of the number of voters), the expected distortion is upper-bounded by $(1+2\alpha)\done^*$, where $\done^*$ is the distortion of the expected winner.

Finally, we generalize our results to include arbitrary metric spaces. 
We show that the same upper bounds obtained for the distortion value for the line metric also work for any arbitrary metric space.

\subsection{Related Work}
The utilitarian view, which assumes that the voters have costs for each alternative, is a well-known approach in welfare economics \cite{roemer1998theories,ng1997case} and has received attention from the AI community during the past decade \cite{procaccia2006distortion,boutilier2015optimal,branzei2013bad,pivato2016asymptotic,anshelevich2017randomized,caragiannis2017subset,goel2017metric,gross2017vote,amanatidis2019peeking,caragiannis2011voting}. Procaccia and Rosenschein \cite{procaccia2006distortion} first introduced distortion as a benchmark for measuring the efficiency of a social choice rule in utilitarian settings. The worst-case distortion of many social choice functions is shown to be high or even unbounded.  However, imposing some mild constraints on the cost functions yields strong positive results. One of these assumptions  which is reasonable in many political and social settings, is the spatial assumption which assumes that the agent costs form a metric space   \cite{enelow1984spatial,merrill1999unified,feldman2016voting,anshelevich2015approximating,anshelevich2016ordinal,pierczynski2019approval,munagala2019improved}. 

Anshelevich, Bhardwaj and Postl \cite{anshelevich2015approximating} were first to analyze the distortion of ordinal social choice functions when evaluated for metric preferences. For plurality and Borda rules, they prove that the worst-case distortion is $2m-1$, where $m$ is the number of alternatives. On the positive side, they show that for the Copeland rule, the distortion value is at most 5. They also prove the lower bound of 3 for any deterministic voting mechanism and conjecture that the worst-case distortion of Ranked Pairs social choice rule meets this lower-bound. This conjecture is later refuted by Goel,  Krishnaswamy, and Munagala \cite{goel2017metric}. Recently, Munagala and Wang \cite{munagala2019improved} present a weighted tournament rule with distortion of $4.236$.

 In addition to deterministic social choice rules, the distortion of randomized rules have been also studied in the literature.  The output of such mechanisms is a probability distribution over the set of alternatives rather than a single winning alternative. 
 Anshelevich and Postl \cite{anshelevich2017randomized}  show that for $\alpha$-decisive metric spaces \footnote{ In an $\alpha$-decisive metric, for every voter, the cost of her preferred choice is at most $\alpha$ times the cost of her second best choice.} any randomized rule has a lower-bound of $1+\alpha$ on the distortion value. For the case of two alternatives, they propose an optimal algorithm with the expected distortion of at most $1+\alpha$. Cheng et al. \cite{cheng2017distortion} characterized the positional voting rules with constant expected distortion value (independent of the number of candidates and the metric space). 
 

 Chen, Dughmi, and Kempe \cite{cheng2017people} consider the case that candidates are drawn randomly from the population of voters. 
They prove the tight bound of $1.1716$ for the distortion value in the line metric and an upper-bound of $2$ for an arbitrary metric space.

In addition to the studies mentioned in Section \ref{abstention}, there are many other studies that consider the effect of abstention in various types of elections. For example,  Desmedt and Elkind in \cite{desmedt2010equilibria} propose a game theoretic analysis of the plurality voting with the possibility of abstention and characterize the preference profiles that admit a pure Nash equilibrium. Rabinovich et al. \cite{rabinovich2015analysis} consider the computational aspects of iterative plurality voting with abstention.  Also, related to our work is the concept of \emph{embedding into voting rules} introduced by  Caragiannis and Procaccia \cite{caragiannis2011voting}.
An embedding is a set of instructions that suggests each agent how to vote, based only on the agent’s own utility function. 
For example, when the voting mechanism is majority, one possible embedding is that voters vote for  each candidate with a probability which is proportion to their utility for that candidate. Among other results, Caragiannis and Procaccia \cite{caragiannis2011voting} show that this embedding results in constant distortion.  
Indeed, our abstention model can be seen as a embedding for elections with majority rule where voters are allowed to abstain.

\section{Preliminaries}
\label{model}

In our study, every election $\election$ consists of four ingredients:
\begin{itemize}
	\item A set $\voters_\election$ of $n$ voters. We denote the $i$'th voter by $\voter_i$.
	\item A set $\candids_\election$ candidates. In this study, we suppose that there are only two candidates and denote the candidates by  $\canL$ (left candidate) and $\canR$ (right candidate). \footnote{In few cases, we also use $\ell'$ and $r'$ to refer to the left and right candidates.}
	\item A finite metric space $\metric_\election$ where the candidates and the voters are located. Unless explicitly stated otherwise, we suppose that $\metric_\election$ is a line, and $\canL$ and $\canR$ are located respectively  at points $0$ and $1$. 
	In addition, each voter is attributed a value $x_i \in (-\infty,\infty)$ which shows her location on the line. We denote by $d_{i, a}$, the distance between voter $\voter_i$ and alternative $ a \in \{\canL,\canR\}$.  
	\item A mechanism by which the winner is selected. In this paper, we consider a simple scenario where the winning candidate is elected via the majority rule (in case of a tie, the winner is determined by tossing a fair coin). Note that for two candidates, almost all the well-known deterministic voting mechanisms select the candidate preferred by the majority as winner.  
\end{itemize}


\begin{definition}
	For an election $\election$ and candidate $a \in \{\canL,\canR\}$, we define the social cost of $a$ in $\election$ as
	$$
	\cost_\election(a) = \sum_{\voter_i \in \voters_\election} d_{i,a}.
	$$
	The optimal candidate of election $\election$, denoted by $o$ is the candidate that minimizes the social cost, i.e.,
	$$
	o_\election = \arg \min_{a \in \{\canL,\canR\}} \cost_\election(a).
	$$
\end{definition}


We suppose that each voter either abstains or votes for one of the candidates. In Section \ref{vb} we give a formal description of the voting behavior of the agents.

\subsection{Voting Behavior of Individuals}\label{vb}
We employ a simple probabilistic model, where each voter independently decides whether to abstain or participate by evaluating her distances from the candidates. Fix an election $\election$ and  a voter $\voter_i \in \voters_\election$ and let $a \in \{\canL,\canR\}$ be the candidate closer to $\voter_i$ in $\metric_\election$ and $\bar{a}$ be the other candidate. We suppose that $\voter_i$ votes sincerely for her preferred candidate $a$ with a probability $p_i$
where $p_i$ is a function of $d_{i,a}$ and $d_{i,\bar{a}}$, and abstains with probability $1-p_i$.  

Denote by $f$ the probability function from which  $p_i$ is derived, i.e., $p_i = f(d_{i,a}, d_{i,\bar{a}})$ 
. Since $f$ represents the probability of voting, we expect $f$ to satisfy certain axiomatic assumptions. Recall that in  spatial voting models, there are two crucial sources of abstention \cite{kirchgassner2003abstention}:  

\begin{itemize}
	\item \textbf{Indifference-based Abstention (IA):} the smaller the difference between the distances of a voter from the candidates is, the less likely it is that she casts a vote.
	
	\item \textbf{Alienation-based Abstention (AA):} the further a voter is located from his preferred candidate, the less likely it is that she casts a vote. 
	
\end{itemize}

To illustrate, for the voters in Figure \ref{CE}, we have:

\begin{itemize}
	\item Voters $\voter_1,\voter_2,$ and $\voter_3$ prefer $\canL$ and voters $\voter_5$ and $\voter_6$ prefer $\canR$. 
	\item Voter $\voter_1$  has a strong incentive to cast a vote since her cost for $\canL$ is zero. 
	\item Voter $\voter_4$ always abstains, since her costs for both the candidates are equal (IA).
	\item For voters $\voter_5$ and $\voter_6$, we have $p_5\geq p_6$, since $\voter_6$ is more alienated (AA).
	\item For voters $\voter_2$ and $\voter_3$, we have $p_2\!\geq\! p_3$, since $d_{2,\canL}\! \leq\! d_{3,\canL}$, and  $d_{2,\canR}\!-\!d_{2,\canL} \!\geq\! d_{3,\canR}\!-\!d_{3,\canL}$ (IA,AA).  

\end{itemize} 

\begin{figure}
	\centerline{\includegraphics{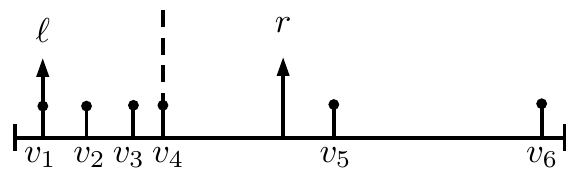}}
	\caption{A simple election.}
	\label{CE}
\end{figure}
As mentioned, the models of Downs \cite{downs1957economic} and Riker and Ordeshook \cite{riker1968theory} are only capable of explaining the Indifference-based abstention, since they only consider the absolute difference between the distances of the candidates to a voter. To resolve this, some recent studies argue that the relative distance, rather than absolute distance, is relevant. 
In this study, we follow the model of Kirchg{\"a}ssner \cite{kirchgassner2003abstention} which is based on the relative distances. The idea is that the probability that a voter casts a vote depends on her ability to distinguish between the candidates. By Weber–Fechner's law
(see \cite{fechner2012elemente}), the ability to distinguish between the candidates depends on their relative distances to the voter. Formally, the probability $p_i$ that voter $\voter_i$ votes for $a$ is calculated via the following formula: 
\begin{equation}
\label{abstentionf}
p_i = f(d_{i,a},d_{i,\bar{a}}) = \frac{|d_{i,a} - d_{i, \bar{a}}|}{d_{i,a} + d_{i, \bar{a}}}.
\end{equation}
Here we consider a more general form of Equation \eqref{abstentionf}. We suppose that each voter $\voter_i$ in election $\election$ casts a vote with probability $p_i$, where
\begin{equation}
p_i = {f}_{\beta} (d_{i,a},d_{i,\bar{a}}) = \left( \frac{|d_{i,a} - d_{i, \bar{a}}|}{d_{i,\bar{a}} + d_{i,a}}\right)^{\beta},
\end{equation}
where $\beta$ is a constant in $[0,1]$.
Figure \ref{betaf} shows the behavior of ${f}_{\beta}$ for different values of $\beta$ and different locations on the line. As is clear from Figure \ref{betaf}, for the smaller values of $\beta$, voters are more eager to participate. Indeed, the exponent $\beta$ can be seen as a quantitative measure of how much this ideological distance matters. For the special case of $\beta=0$, voters always participate in the election, regardless of their location. We refer to $\beta$ as the participation parameter. It can be easily observed that for any $0 \leq \beta \leq 1$, function ${f}_{\beta}$ satisfies all the desired criteria.

\begin{figure}
	\centerline{\includegraphics[scale=0.28]{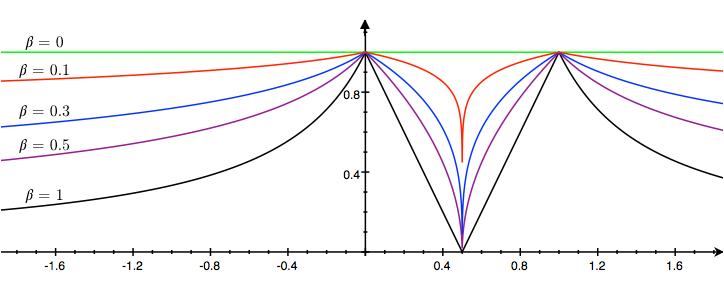}}
	\caption{${f}_{\beta}$ for different values of $\beta$, when two candidates $\canL$ and $\canR$ are located at points $0$ and $1$. For any point $z$, the curves show $f_{\beta}(|z-0|,|z-1|)$ for different $\beta$.}
	\label{betaf}
\end{figure}

\subsection{ Expected Winner and  Expected Distortion}

As mentioned, our assumption is that the winner is determined by the  majority rule.  However, according to the stochastic behavior of the voters, the winner is not deterministic, i.e., each candidate has a probability of winning. Denote by ${\small \texttt {\#}}{a}_\beta$, the expected number of voters who vote for candidate $a$, when the participation parameter is $\beta$. Furthermore, denote by $\mathtt{P}_{a,\beta}$, the probability that candidate $a$ wins the election, when the participation probability is $\beta$. 
We define the \textbf{expected winner} of  election $\election$ for participation parameter $\beta$, denoted by $\omega_{\election,\beta}$ as the candidate with the maximum expected number of votes. 
$$\omega_{\election,\beta} =  \arg \max_{a \in \candids_\election} { \small \texttt {\#}}{a}_\beta.$$

\begin{definition}
	For election $\election$ and a candidate $a \in \{\canL,\canR\}$ we define the distortion of $a $ in election $\election$, denoted by $D(a)$, as the ratio
	$
	\cost_\election(a) / \cost_\election(o_\election).
	$
\end{definition}

By definition, the distortion of the optimal candidate is $1$.
We discuss two approaches to evaluate the distortion of an election $\election$. In the first approach, we evaluate election $\election$ by the distortion of its expected winner, i.e., $D(\omega_{\election,\beta})$. 
  Another approach  is to define the distortion of election $\election$ as the expected distortion of the winner, over all possible outcomes, i.e.,
\begin{equation}
\label{eddef}
\dtwo_\beta(\election) = \mathtt{P}_{\ell,\beta} \cdot D(\canL) + \mathtt{P}_{r,\beta} \cdot D(\canR). 
\end{equation}
Finally, for any $0 \leq \beta \leq 1$, we define worst-case distortion values ${D}^*_\beta$ and $\dtwo^*_\beta$ as: 
$$
\done^*_\beta = \max_{\election \in \Omega} D(\omega_{\election,\beta}), $$ 
and 
$$
\dtwo^*_\beta = \max_{\election \in \Omega} \dtwo_\beta(\election),
$$
where $\Omega$ is the set of all possible elections $\election$.
We dedicate two separate sections to analyze  the value of  $
\done^*_\beta  $ and  
$
\dtwo^*_\beta$.
Even though the value of $\done^*_\beta $ and  $\dtwo^*_\beta$ essentially depend on $\beta$, we provide necessary tools to analyze  distortion these values for any $\beta \in [0,1]$.

For convenience, in the rest of the paper, when $\beta$ is  fixed we drop  the subscript `$\beta$' and simply use $ \probL,\shX $ instead of $\mathtt{P}_{\ell,\beta}, {\small \texttt {\#}}{a}_\beta .$
\section{Distortion of the Expected Winner}
\label{EW}
Throughout this section, we analyze the worst-case  distortion of the expected winner. Recall that the expected winner is the candidate with a higher expected number of votes. 
There are two reasons why we consider the distortion value of the expected winner. First, since the number of votes that each candidate receives is concentrated around its expectation \footnote{We can show this claim using concentration bounds such as \emph{Hoeffding}. A simple form of this inequality states that for $n$ independent random variables bounded by $[0,1]$, we have
$$
\mathbb{P}(S_n - \mathbb{E}[S_n]>t ) \leq \exp(-2nt^2),
$$
where $S_n$ is the sum of the variables. 
 }, in elections with a large number of voters, the expected winner has a very high chance of winning; especially when there is a non-negligible separation between the expected number of votes that each candidate receives.
Secondly, we use the tight upper-bound on the distortion value of the expected winner to prove an upper bound on the expected distortion of the election for the second approach.
Recall that the probability that a voter $\voter_i$ casts a vote for his favorite candidate in election $\election$ is:  
$$
{f}_{\beta} =  \left(\frac{|d_{i,\canL} - d_{i, \canR}|}{d_{i,\canL} + d_{i, \canR}}\right)^{\beta}.
$$ 
In this section, we suppose without loss of generality that candidate $\canL$ is the expected winner. Moreover, we assume that the optimal candidate is $\canR$; otherwise the distortion  equals $1$. 
We also consider four regions $\regA,\regB,\regC$ and $\regD$ as in Figure \ref{fig:line}.

\begin{figure}[h]
	\begin{center}
		\includegraphics{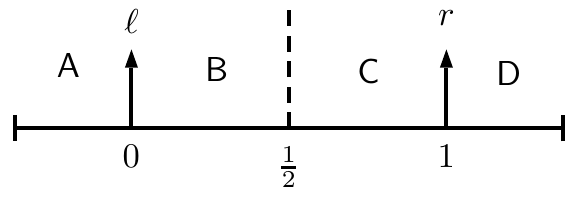}
	\end{center}
	\caption{Regions $\regA,\regB,\regC,$ and $\regD$}
	\label{fig:line}
\end{figure}

In Theorem \ref{theorem1} we state the main result of this section. 

\begin{theorem}
	\label{theorem1}
	For any $\beta \in [0,1]$, 
	there exists an election $\election$, such that $D(\omega_{\election,\beta}) = \done^*_\beta$ 
	and the voters in $\election$ are located at  two different locations $x_b \in \regB$ and $x_d \in \regD$. 
\end{theorem}

The basic idea to prove Theorem \ref{theorem1} is as follows: we prove that for every election $\election$, there exists an election $\election'$ with $\done(\omega_{\election',\beta}) \geq \done(\omega_{\election,\beta})$, such that the voters in $\election'$ are located in at most $2$ different locations. To show this, we collect the voters in $\election$ by carefully moving them forward and backward via a sequence of \emph{valid displacements}, as defined in Definition \ref{validd}.

\begin{definition}
	Define a displacement as the operation of moving a subset of the voters forward or backward on the line to a new location. A displacement is valid if it does not alter the expected winner, and furthermore, does not decrease the distortion value of the expected winner.
	\label{validd}
\end{definition}

In Lemmas \ref{lemma1}, \ref{lembc}, and \ref{lem3} we introduce three sorts of valid displacement which help us collect the voters. For convenience, here we only state the lemmas and  defer the proofs  to Section \ref{validdisp}. Figure \ref{displacements}, illustrates a summary of the valid displacements introduced in these lemmas. Note that these displacements are valid for any  $\beta \in [0,1]$.

\begin{restatable}{lemma}{firstlemma}
	\label{lemma1}
	Moving a voter $\voter_i$ from $x_i \in \regA$ to $0$ is a valid displacement.
\end{restatable}

\begin{restatable}{lemma}{secondlemma}
	\label{lembc}
	Consider voters $\voter_i$ and $\voter_j$  respectively at $x_i \in \regB$ and $x_j \in \regC$. Then,
	\begin{itemize}
		\item If $d_{i,\canL} \leq d_{j,\canR}$, moving $\voter_i$ to $x_i+ x_j-1/2$ and $\voter_j$ to $1/2$ is a valid displacement. 
		\item If $d_{i,\canL} > d_{j,\canR}$, moving $\voter_i$ to $x_i - 1 + x_j$ and $\voter_j$ to $1$ is a valid displacement.
	\end{itemize}
\end{restatable}

\begin{restatable}{lemma}{thirdlemma}
	\label{lem3}
	Consider voters $\voter_i, \voter_j$, where $x_i,x_j \in \regB$ or $x_i,x_j \in \regD$. Then moving both the voters to $(x_i+ x_j)/2$ is a valid displacement.
\end{restatable}

\begin{figure}[h]
	\begin{center}
		\begin{tikzpicture}
		\node[draw]{\includegraphics[scale=1.2]{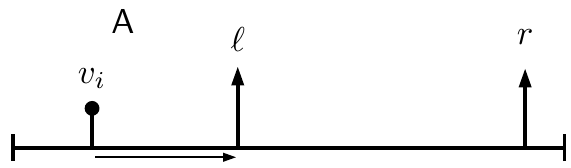}};
		\end{tikzpicture}
		
		\vspace{1mm}
		\begin{tikzpicture}
		\node[draw]{\includegraphics[scale=1.2]{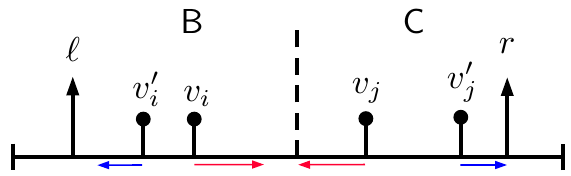}};
		\end{tikzpicture}
		
		\vspace{1mm}
		\begin{tikzpicture}
		\node[draw]{\includegraphics[scale=1.2]{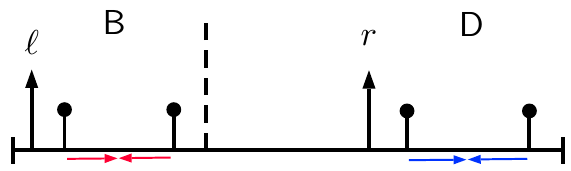}};
		\end{tikzpicture}
		
		\caption{Valid displacements introduced in Lemmas 
			\ref{lemma1}, \ref{lembc}, \ref{lem3}.}
		\label{displacements}
	\end{center}
\end{figure}
We also state two simple and natural Corollaries of Lemmas \ref{lembc}, and \ref{lem3}.
	
	\begin{corollary}[of Lemma \ref{lembc}]
		\label{ColBC}
		We can move each voter in region $\regC$ to either $1$ or $1/2$ by a sequence of valid displacements.
	\end{corollary}
	
	\begin{proof}
		Consider an arbitrary voter $\voter_j \in \left(1/2,1\right)$. Since $\canL$ is the expected winner, there exists at least one voter, say $\voter_i$, in region $\regB$. By Lemma \ref{lembc}, if $d_{i,\canL} \leq d_{j,\canR}$, we can move $\voter_j$ to $1/2$ and if $d_{i,\canL} > d_{j,\canR}$, we can move $v_j$ to $1$. 
	\end{proof}
	
	\begin{corollary}[of Lemma \ref{lem3}]
		\label{ColBtop}
		We can collect all the voters of region $\regB$ at some point $x \in \regB$ via  a sequence of valid displacements.
		Furthermore, we can collect all the voters of region $\regD$ at some point $x' \in \regD$ via a sequence of valid displacements.
	\end{corollary}
	
	\begin{proof}
		By applying Lemma \ref{lem3} iteratively to the furthest voters, the maximum distance between the voters in each region decreases. This procedure can  be applied until all the voters gather at one point.
	\end{proof}
	
	Now, we are ready to prove Theorem \ref{theorem1}.
	\begin{proof}[Proof of Theorem \ref{theorem1}.]
	First, we prove that by Lemma \ref{lemma1}  and Corollaries \ref{ColBC}, and \ref{ColBtop}, every election $\election$ can be reduced to an election $\election'$, such that 
	\begin{itemize}
		\item $\election$ and $\election'$ have the same expected winner. 
		\item $\done\left(\omega_{\election,\beta}\right) \leq \done\left(\omega_{\election',\beta}\right)$. 
		\item All the agents in $\election'$ are located at two points $x_b \in \regB$ and  $x_d \in \regD$.
	\end{itemize}
	Consider an arbitrary election $\election$. Using Lemma \ref{lemma1}, we move all the voters in region $\regA$ to $0$. Afterwards, using  Corollary \ref{ColBC} we move each voter in region $\regC$ to one of the points $1/2$ or $1$. At this point, all the voters belong to one of regions $\regB$ or $\regD$ (we suppose that the voters located in the borderlines belong to both regions). Finally, using Corollary \ref{ColBtop}, we collect all the voters in regions $\regB$ and $\regD$ at some points $x_b \in \regB,$ $x_d \in \regD$. 
	
	Finally, let $\election$ be an arbitrary election such that $\done(\omega_{\election,\beta}) = \done_\beta^*$. 
	 Applying the above reduction on $\election$, yields an election $\election^{*}$ with  $\done(\omega_{\election^*,\beta}) = \done_\beta^*$, and the desired structure.
\end{proof}       



According to \cref{theorem1}, for any $\beta \in [0,1]$,  we can establish an election $\election^*$ with the maximum distortion, and the following structure (see Figure \ref{fig:final}): the interior of regions $\regA$ and $\regC$ contain no voter. All the voters are located at two points   $x_b \in \regB$ and $x_d \in \regD$. Note that, the maximum distortion value and the location of $x_b$ and $x_d$ in the worst-case scenario depends on the value of $\beta$. 

\begin{figure}[h]
	\begin{center}
		\includegraphics{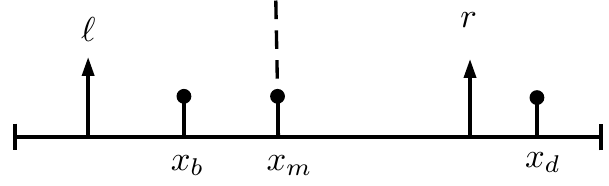}
	\end{center}
	\caption{For any $\beta \in [0,1]$, there is an election $\election^*$ with $\done(\omega_{\election^*,\beta}) = \done_\beta^*$, and the above structure.}
	\label{fig:final}
\end{figure}

\subsection{A Tight Upper Bound on $\done^*_\beta$} \label{ubond}
We now evaluate  $\done^*_\beta$ for different values of $0\leq \beta \leq 1$.
Let us start by the boundary case $\beta = 0$.  For $\beta = 0$, the probability that a voter casts a vote is $1,$
independent of her location. It is proved that for this case, we have  $\done^*_\beta=3$ \cite{anshelevich2015approximating}. Indeed, the same example we provided in Figure \ref{3app} is the scenario with the highest distortion for $\beta=0$.   


Now, consider $\beta>0$, and let $\election^*$ be the election that maximizes $\done(\omega_{\election^*,\beta})$. As discussed in the previous section, we can assume without loss of generality that the voters in $\election^*$ are located at two points, namely, $x_b \in \regB$ and $x_d \in \regD$. Suppose that $q_b$ voters are at $x_b$ and $q_d$ voters are at $x_d$. We have: 
$$\shL = \left(1-2x_b \right)^\beta q_b \quad \mbox{and}\quad \shR = \left(\frac{1}{(2x_d-1)^\beta}\right)q_d.$$
Since $\canL$ is the expected winner, we have 

$$
(1-2x_b)^\beta q_b \geq \left(\frac{1}{(2x_d-1)^\beta}\right)q_d.
$$ 
On the other hand, we have 

$$\cost_{\election^*}(\canL) = q_bx_b + q_dx_d ,$$
and 

$$\cost_{\election^*}(\canR) = q_b(1-x_b) + q_d(x_d - 1).$$
Thus, 

\begin{align*}
\done(\canL)     &= \frac{\cost_{\election^*}(\canL)}{\cost_{\election^*}(\canR)}\\
&= \frac{q_bx_b + q_dx_d}{q_b(1-x_b) + q_d(x_d - 1)}\\
&= \frac{q_b x_b  + (n-q_b) x_d }{q_b(1-x_b) + (n-q_b)(x_d - 1)} 
\end{align*}
Therefore, in order to find the maximum distortion value, we need to solve the following  optimization problem:

\begin{maxi}
	{}{ \frac{q_b x_b  + (n-q_b) x_d }{q_b(1-x_b) + (n-q_b)(x_d - 1)}}
	{}{}
	\addConstraint{(1-2x_b)^\beta q_b }{\geq \frac{n-q_b}{(2x_d-1)^\beta}}
	\addConstraint{0 \leq q_b}{\leq 1} 
	\addConstraint{0 \leq x_b}{ \leq 1/2}
	\addConstraint{1 }{\leq x_d.}   
	\label{cp2}
\end{maxi}
Now consider another boundary case: $\beta = 1$. For $\beta = 1$
the answer to the above optimization  problem is $\frac{(1+\sqrt2)^2}{1+2\sqrt2} \simeq 1.522$, which can be obtained by choosing $q_b=\frac{n}{2+\sqrt2}$,  $x_b=0$, and $x_d = \frac{2+\sqrt2}{2}$. A graphical representation of this construction is shown in Figure \ref{fig:ans}.


\begin{figure}[h]
	\begin{center}
		\includegraphics{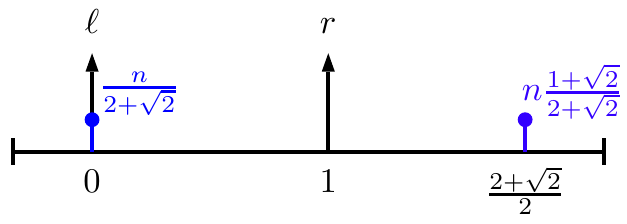}
	\end{center}
	\caption{A tight example for $\beta = 1$.}
	\label{fig:ans}
\end{figure}

In general for $0 < \beta < 1$, the maximum distortion value equals the answer of Optimization Problem (\ref{cp2}). In Figure \ref{fig:beta}, we show the answer of this program for different values of $\beta$. Interestingly, with $\beta$ increasing from $0$ to $1$, $\done_\beta^*$  initially decreases and then increases. As illustrated in Figure \ref{fig:beta}, it can be seen that the minimum possible value for $\done^*_\beta$ is $\simeq \sqrt{2}$ for $\beta \simeq 0.705$. 

\begin{figure}[h]
	\begin{center}
		\includegraphics[scale=0.2]{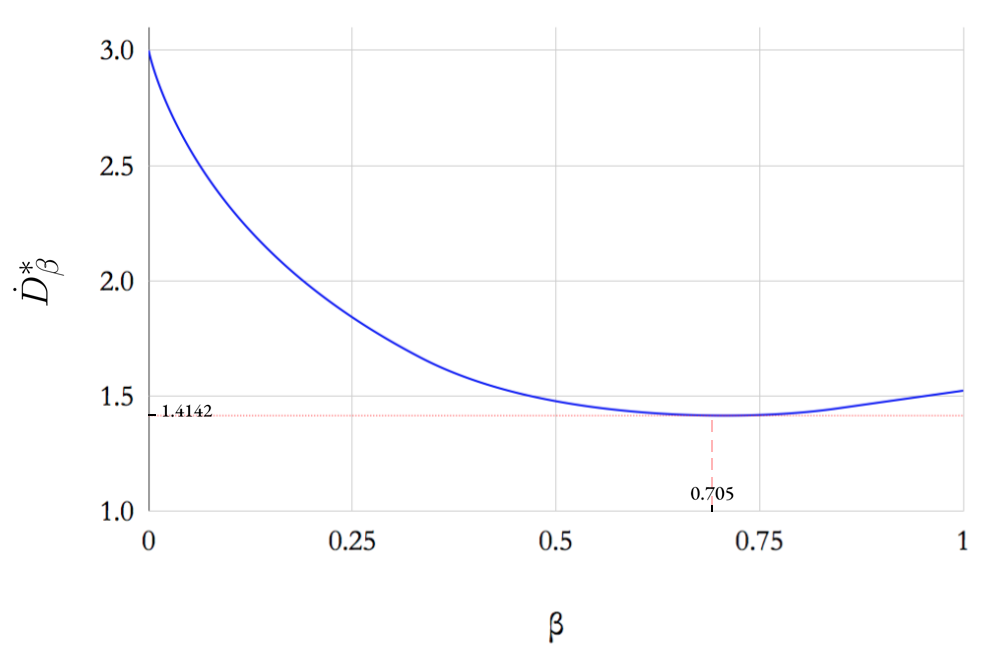}
	\end{center}
	\caption{Worst-case distortion for $0\! \leq\! \beta\! \leq\! 1$.}
	\label{fig:beta}
\end{figure}

\subsection{Valid Displacements} \label{validdisp}
In this section, we prove Lemmas \ref{lemma1}, \ref{lembc}, and \ref{lem3}.
One important tool to prove these lemmas is  Observation \ref{lem:frac}.

\begin{observation}
	Let $a,b,c,d>0$ be four positive constants. We have :\footnote{For the second inequalities, we assume  $d<b$.}
	\begin{itemize}
		\item If ${a \over b} > \frac{c}{d}$ then ${a+c \over b+d} < {a \over b}$ and ${a-c \over b-d} > {a \over b} $.
		\item  If ${a \over b} < \frac{c}{d}$ then ${a+c \over b+d} > {a \over b}$ and ${a-c \over b-d} < {a \over b} $.
	\end{itemize}
	\label{lem:frac}
\end{observation}

\firstlemma*
\begin{proof}
	Initially,
	$\voter_i$ votes for $\canL$ with probability $\left(\frac{1}{1-2x_i}\right)^\beta$. After moving $\voter_i$ to $0$, she votes for $\canL$ with probability
	$1$. Therefore,  if we move $\voter_i$ to $0$, the value of  $\shL$ does not decrease, and the expected winner does not change.
	Furthermore, by this movement both $\cost_{\election}\left(\canL\right)$ and $\cost_{\election}\left(\canR\right)$ are decreased by $-x_i$. Let $c$ and $c'$ be the contribution of  $\voter_{-i}$ (that is, all voters except $v_i$) to the social cost of $\canL$ and $\canR$ respectively. Before moving $\voter_i$ to $0$, we have
	
	$$\done(\ell) = {c -x_i \over c' + {1-x_i}},$$ 
	and after the movement we have
	\begin{align}
	\done(\ell) &=  {c  \over c' + {1}} \nonumber \\
	&= {(c - x_i )- (-x_i) \over (c' + 1-x_i) - (-x_i)},\label{nmm}
	\end{align}
	 By applying  Observation \ref{lem:frac} on Equation \eqref{nmm}, we have
	$$
	 {c -x_i \over c' + {1-x_i}} \leq {(c - x_i )- (-x_i) \over (c' + 1-x_i) - (-x_i)}  
	$$
	which implies that moving $\voter_i$ to $0$ is a valid displacement.
\end{proof}

\secondlemma*

\begin{proof}
	Initially, $\voter_i$ votes for $\canL$ with probability $\left(1-2x_i\right)^\beta$, and  $\voter_j$ votes for $\canR$ with probability $\left(2x_j-1\right)^\beta$. Since these movements do not change the regions where the voters belong, after the movement they vote for the same candidate but with different probabilities. Let $\Delta_i$ be the difference between the contribution of $\voter_i$ to $\shL$, before and after the displacement. Similarly, let $\Delta_j$ be the difference between the contribution of $\voter_j$ to $\shR$ before and after the movement. We consider two cases.
	
	\vspace{0.2cm}
	\textbf{Case I} ($d_{i,\canL} \leq d_{j,\canR}$): if we move $\voter_i$ to $x_i + x_j - 1/2$ and $\voter_j$ to $1/2$,  $\voter_i$ votes for $\canL$ with probability
	$\left(2-2x_i-2x_j\right)^\beta$
	and $\voter_j$ votes for $\canR$ with probability 
	$0$. Thus, 
	we have
	$$\Delta_i = \left(2 - 2x_j-2x_i \right)^\beta- \left(1-2x_i\right)^\beta
	\qquad \mbox{and} \qquad
	\Delta_j =0^\beta - \left(2x_j-1\right)^\beta.$$
	Since $\beta \leq 1$, by straightforward calculus we have:
	\begin{align*}
	\left((1 - 2x_i) - (2x_j -1 ) \right)^\beta &\geq  \left(1-2x_i\right)^\beta - \left(2x_j - 1\right)^\beta\\
	\left((1 - 2x_i) - (2x_j -1 ) \right)^\beta - \left(1-2x_i\right)^\beta &\geq  - \left(2x_j - 1\right)^\beta\\
	\Delta_i &\geq \Delta_j.
	\end{align*}
	
	\textbf{Case II} ($d_{i,\canL} > d_{j,\canR}$): if we move $\voter_i$ to $x_i + x_j - 1$ and $\voter_j$ to $1$, after the displacement, $\voter_i$ votes for $\canL$ with probability
	$\left(3-2x_i-2x_j\right)^\beta,$
	and $\voter_j$ votes for $\canR$ with probability 
	$1$.
	Therefore, we have
	$$\Delta_i = \left(3 - 2x_i-2x_j \right)^\beta- \left(1-2x_i\right)^\beta 
	\qquad \mbox{and} \qquad 
	\Delta_j =1^\beta - \left(2x_j-1\right)^\beta.$$
	Since $\beta \leq 1$, we have
	\begin{align*}
	\left((1 - 2x_i) - (2x_j -2 ) \right)^\beta &\geq  \left(1-2x_i\right)^\beta - \left(2x_j - 2\right)^\beta\\
	\left((1 - 2x_i) - (2x_j -2 ) \right)^\beta - \left(1-2x_i\right)^\beta &\geq  - \left(2x_j - 2\right)^\beta\\ 	\left((1 - 2x_i) - (2x_j -2 ) \right)^\beta - \left(1-2x_i\right)^\beta 
	& \geq 1^\beta - (2x_j - 1)^\beta\\
	\Delta_i &\geq \Delta_j.
	\end{align*}
	Hence the expected winner does not change.
	In addition, since we move two voters in Regions $\regB$ and $\regC$ equally in the opposite directions in both cases, the distortion value of each candidate remains unchanged.
\end{proof}

\thirdlemma*

\begin{proof}
	Let $\varepsilon = |x_i-x_j|/2$. Recall the definition of $\Delta_i$ and $\Delta_j$ from the proof of Lemma \ref{lembc}. For the case of $x_i,x_j \in \regD$, we have:
	
	\begin{align*}
	\Delta_i &= \left(\frac{1}{2x_i+2\varepsilon-1}\right)^\beta - \left(\frac{1}{2x_i-1}\right)^\beta  ,
	\end{align*}
	and
	\begin{align*}
	\Delta_j &= \left(\frac{1}{2x_j-2\varepsilon-1}\right)^\beta - \left(\frac{1}{2x_j-1}\right)^\beta .
	\end{align*}
	Thus, we have  
	$$
	\Delta_i+ \Delta_j = \left(\frac{1}{2x_i+2\varepsilon-1}\right)^\beta - \left(\frac{1}{2x_i-1}\right)^\beta +\left(\frac{1}{2x_j-2\varepsilon-1}\right)^\beta - \left(\frac{1}{2x_j-1}\right)^\beta.
	\label{sumdelta}
	$$
	Since $x_j>x_i$,  these two inequalities imply $\Delta_i + \Delta_j \leq 0$ (see Figure \ref{decreasing-convex}). Thus, value of $\shR$ does not increase and the expected winner does not change.	
	
	Similarly For the case of $x_i,x_j \in \regB$, we have:
	\begin{align*}
		\Delta_i &= \left(1-2x_i-2\varepsilon\right)^\beta - \left(1-2x_i\right)^\beta  ,
	\end{align*}
	and
	\begin{align*}
		\Delta_j &= \left(1-2x_j+2\varepsilon\right)^\beta - \left(1-2x_j\right)^\beta .
	\end{align*}
	Thus, we have  
	$$
	\Delta_i+ \Delta_j = \left(1-2x_i-2\varepsilon\right)^\beta - \left(1-2x_i\right)^\beta + \left(1-2x_j+2\varepsilon\right)^\beta - \left(1-2x_j\right)^\beta .
	$$
	Note that since ${f'}_\beta\left(x\right) = \left(1-2x_i\right)^\beta$ is a decreasing and concave function, we have $$\frac{d\left({f'}_\beta\right)}{dx} \leq 0,$$ and $$\frac{d^2\left({f'}_\beta\right)}{dx^2} \leq 0.$$
	Since $x_j>x_i$,  these two inequalities imply $\Delta_i + \Delta_j \geq 0$. Thus,  value of $\shL$ does not decrease and the expected winner does not change.
	In addition, since the voters move in the opposite directions and by the same distance, the distortion value of the candidates do not change. Therefore, this modification is a valid displacement.
	\begin{figure}
		\centerline{\includegraphics[scale=0.35]{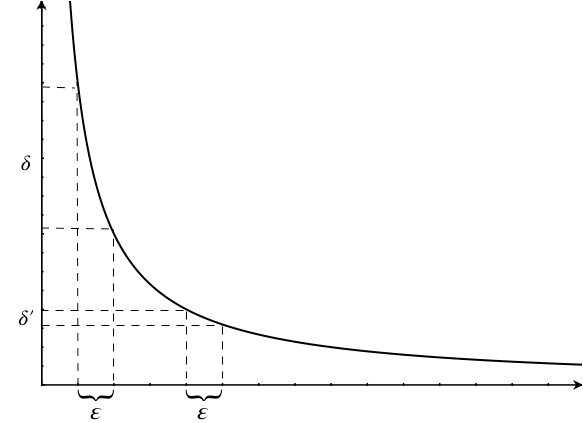}}
		\caption{For every decreasing convex function $g$ and $x_i < x_j$, we have     $g\left(x_i + \varepsilon\right) - g\left(x_i\right)  \leq g\left(x_j\right)-g\left(x_j-\varepsilon\right)  \leq 0$}
		\label{decreasing-convex}
	\end{figure}
\end{proof}
\section{Expected Distortion}
\label{ED}
Recall that in our second approach, we define the distortion of an election as the expected distortion of the winner, where the expectation is taken over the random behavior of the voters. Our main result in this Section is Theorem \ref{bbnd}.

\begin{restatable}{theorem}{secondthm}
	\label{bbnd}
	For any $\alpha>0$, value of $\dtwo_\beta(\election)$ for every election $\election$ whose candidates receive at least
		\[\phi(\alpha) = \frac{(\alpha + 1)^3}{\alpha^2 (\alpha - \sqrt{\alpha + 1})^2}\]
	expected number of votes is at most $(1+2\alpha)\done^*_{\beta}$.
\end{restatable}

In this section, we suppose without loss of generality that candidate $\canR$ is the optimal candidate. Thus, Equation \eqref{eddef} can be rewritten as 

\begin{equation}
\label{secondform}
\dtwo_\beta(\election) = \probL \frac{\cost_{\election}(\canL)}{\cost_{\election} (\canR)} + \probR.
\end{equation}

	In this case, if $\canL$ would be the expected winner, we have:
\begin{align}
\dtwo_\beta(\election) 	&= \probL D(\canL) + \probR D(\canR )\nonumber\\
&\leq  \probL D(\canL)+ \probR D(\canL)		& (D(\canL) \geq D(\canR)) \nonumber\\
&= D(\canL). \label{temp3}
\end{align}
In addition, we know that the distortion of the expected winner is at most $\done^*_{\beta}$, which together with Equation \eqref{temp3} implies $\dtwo(\election) \leq \done^*_{\beta}$ for the case that $\canL$ is the expected winner. Therefore, throughout this section we suppose that $\canR$ is both the optimal and the expected winner candidate.

In Theorem \ref{theorem2}, we prove that there is an election with the maximum distortion value and a simple structure.

\begin{theorem}
	\label{theorem2}
	For any $\beta \in [0,1]$, there exists an election $\election^*$  such that $\dtwo_\beta(\election^*)$ is maximum, and in $\election^*$ there is no voter in the interior of regions $\regA$ and $\regC$, and also all the voters in $\regD$ are located at a single point $x_d \in \regD$.  
\end{theorem}

The basic idea to prove Theorem \ref{theorem2} is as follows: we prove that for every election $\election$, there exists an election $\election'$ with $\dtwo_\beta(\election') \geq \dtwo_\beta(\election)$ and the desired structure. To show this, we collect some of the the voters in $\election$ via a sequence of \emph{valid displacements}, albeit with a new definition for valid displacement.

\begin{definition}
	A displacement is valid, if it does not decrease $\dtwo(\election)$.
\end{definition}

The process of proving that a displacement is valid for this case is relatively tougher than the previous model. The reason is that we do not even  have a closed-form expression which represents the winning probability of each candidate.  In Lemmas \ref{vd1} and \ref{toward}  we explain our tools to discover valid displacements. For brevity, we defer the proofs to these lemmas to Section \ref{validdisp2}.

\begin{restatable}{lemma}{vdfirst}
	\label{vd1}
	For each voter $\voter_i  \in \regA$, there is a point $x_i' \in \regB$ such that moving $v_i$ to $x_i'$ is a valid displacement.
	Furthermore, for each voter $\voter_j \in \regC$, there is a point $x_j' \in \regD$ such that moving $v_j$ to $x_j'$ is a valid displacement.
\end{restatable}

\begin{restatable}{lemma}{twd}
	\label{toward}
	Let $\voter_i$ and $\voter_j$ be two voters located respectively at $x_i, x_j \in \regD$. Then, there exists a point $x$ between $x_i$ and $x_j$, such that moving both the voters to $x$ is a valid displacement.
\end{restatable}

\begin{figure}
	\begin{center}
		\begin{tikzpicture}
		\node[draw]{\includegraphics[scale=1.2]{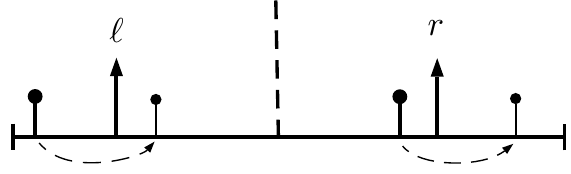}};
		\end{tikzpicture}
	
		\vspace{1mm}
		\begin{tikzpicture}
		\node[draw]{\includegraphics[scale=1.2]{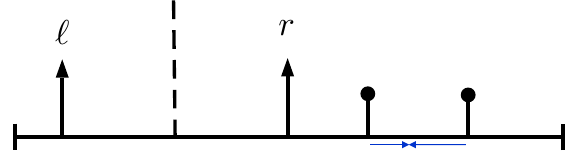}};
		\end{tikzpicture}
		\caption{Valid displacements introduced in Lemmas \ref{vd1} and \ref{toward}.}    
		\label{displacements2}
	\end{center}
\end{figure}

Figure \ref{displacements2},  shows a summary of the displacements described in Lemmas \ref{vd1} and \ref{toward}. Using these displacements, one can establish an election with the maximum expected distortion, and the following structure (see Figure \ref{fig:final-2}): 
the interior of regions $\regA$ and $\regC$ contain no voter. All the voters in $\regD$ are located at point $x_d \in \regD$. 

\begin{figure}[h]
	\begin{center}
		\includegraphics{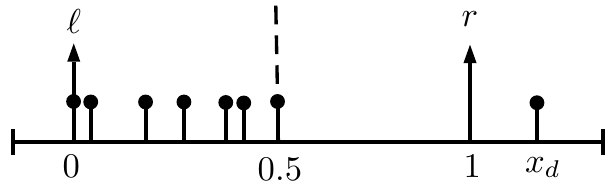}
	\end{center}
	\caption{For any $\beta \in [0,1]$, there is an election with the maximum expected distortion and this structure.}
	\label{fig:final-2}
\end{figure}
\begin{proof}[Proof of Theorem \ref{theorem2}]
	Consider  an election with the maximum expected distortion. By Lemma \ref{vd1} we can suppose that the interior of regions $\regA$ and $\regC$ are empty. Furthermore, by iteratively applying Lemma \ref{toward} on the farthest pair of points in Region $\regD$, we can collect all the voters of $\regD$ into a single point and transform the election into one with the maximum distortion, and the desired structure.

\end{proof}

\subsection{An Almost Tight Bound on $\dtwo^*_\beta$} \label{aatb}
In this section, we discuss  the value of $\dtwo^*_\beta$, for any $\beta \in [0,1]$. As mentioned, to prove our upper and lower bounds in this section, we use the bounds obtained in Section \ref{ubond}. 

Similar to Section \ref{ubond}, we begin with the boundary case of $\beta=0$. By a similar argument as in Section \ref{ubond}, for $\beta=0$ all the voters vote for their preferred candidate and so we have $\dtwo^*_0=3$. For $\beta > 0$, we prove Theorem \ref{bbnd} which provides an asymptotic upper bound on $\dtwo^*_\beta$ for any $\beta \in (0,1]$.


\secondthm*

To prove Theorem \ref{bbnd}, we first prove Lemmas \ref{lbnd} and \ref{fuch}. 

\begin{lemma}
	\label{lbnd}
	Let $\alpha$ be a constant, and $\election$ be an election, with the property that $\canR$ is the optimal and the expected winner candidate, and $\shR/\shL \leq 1+\alpha$. Then $\dtwo_\beta(\election)\leq (1+2\alpha) \done^*_\beta.$
\end{lemma}
\begin{proof}
	To prove this lemma we add sufficient number of agents at point $0$ to alter the expected winner to $\canL$. After this operation, since $\canL$ is the expected winner, we know that the expected distortion of $\canL$ is at most $\done^*_\beta$. Next, based on the number of voters added at point $0$, we bound the value of $\dtwo_\beta(\election)$.  
	
	Let $s$ be the minimum number of voters we need to add at point $0$ to convert $\canL$ to the expected winner. Since $\shR \leq (1+\alpha)\cdot \shL$, and each voter at point $0$ contributes $1$ to $\shL$, we have $ s \leq \alpha \cdot \shL.$ Let $\election'$ be the election, after adding the agents at point $0$. 
	Since the expected winner of $\election'$ is $\canL$, the expected distortion of $\canL$ is upper bounded by $\done^*_\beta$:
	\begin{equation}
	\label{eqeqeq}
	\frac{\cost_{\election'}(\canL)}{\cost_{\election'}(\canR)} \leq \done^*_\beta.
	\end{equation}
	
	Moreover, since we add the agents at point $0$, their cost for candidate $\canL$ is zero and hence, $\cost_{\election'}(\canL) = \cost_{\election}(\canL)$. Thus, we have
	\begin{equation}
	\label{eqeq}
	\frac{\cost_{\election}(\canL)/\cost_{\election}(\canR)}{\cost_{\election'}(\canL)/\cost_{\election'}(\canR)} =\frac{\cost_{\election'}(\canR)}{\cost_{\election}(\canR)}
	\end{equation}
	
	Now, we show that the ratio $\cost_{\election'}(\canR)/\cost_{\election}(\canR)$ is upper bounded by $1 + 2\alpha$. First, let us calculate the explicit formulas of $\cost_{\election}(\canR)$ and $\cost_{\election'}(\canR)$. 
	As discussed before, we can assume 
	that the agents in $\election$ are located either in Region $\regB$ or at point $x_d \in \regD$. Let $q_d$ be the population of the voters that are located at $x_d$. We have 
	$$
	\cost_{\election}(\canR) = \sob (1-x_\voter)  + q_d(x_d-1).
	$$
	Furthermore, we have 
	$
	\cost_{\election'}(\canR) = \cost_{\election}(\canR) + s,
	$
	where 
	\begin{align*}
	s &\leq \alpha \cdot \shL  \\
	&= \alpha \sob(1-2x_\voter)^\beta.
	\end{align*}
	Thus, we have 
	\begin{align*}
	\frac{\cost_{\election'}(\canR)}{\cost_{\election}(\canR)} &\leq 1 + \frac{\alpha \sob(1-2x_\voter)^\beta}{ \sob (1-x_\voter)  + q_d(x_d-1)} \\
	&\leq 1 + \alpha \cdot\frac{\sob(1-2x_\voter)^\beta}{ \sob (1-x_\voter)},
	\end{align*}
	and since for any $x \leq 1/2$ we have $\frac{(1-2x)^\beta}{1-x} \leq 2$, 
	\begin{equation} \label{sx}
\frac{\cost_{\election'}(\canR)}{\cost_{\election}(\canR)} \leq 1+ 2\alpha.
	\end{equation} 
	Inequality \eqref{sx} together with Equations \eqref{eqeqeq} and \eqref{eqeq} implies: 
	\begin{align*}
	\frac{\cost_{\election}(\canL)/\cost_{\election}(\canR)}{\done^*_\beta} \leq 1+ 2\alpha.
	\end{align*}
	Thus, by Equation \eqref{secondform}, we have 
$$\dtwo_\beta(\election) \leq \probL (1+2\alpha) \done^*_\beta+ \probR,$$
and since $\probL  + \probR = 1$ we conclude that
 $\dtwo_\beta(\election) \leq (1+2\alpha) \done^*_\beta.$
\end{proof}

\begin{lemma}
	\label{fuch}
	Let $\alpha$ be a constant, and $\election$ be an election, with the property that $\canR$ is the optimal and the expected winner candidate, and $\shR/\shL > 1+\alpha$. Then, if the number of candidates would be large enough, we have  $\dtwo_\beta(\election)\leq (1+2\alpha) \done^*_\beta$.
\end{lemma}

\begin{proof}
	To prove Lemma \ref{fuch}, we use the fact that the number of votes that a candidate receives is concentrated around it's expected value. By definition, we have 
		\begin{align}
	\dtwo_\beta(\election) &= \probL\frac{\cost_{\election}(\canL)}{\cost_{\election}(\canR)} + (1-\probL)\nonumber\\
	&= \probL (\frac{\cost_{\election}(\canL)}{\cost_{\election}(\canR)}-1) + 1 \nonumber \\
	&= \probL \frac{\sob (2x_\voter-1)+q_d}{\sob(1-x_\voter)+ q_d(x_d-1)} +\! 1 \nonumber\\
	&\leq\probL \frac{q_d}{\sob(1-x_\voter)} + 1  &(x_\voter <1/2) \nonumber\\
	&\leq\probL \frac{q_d}{\sob(1-2x_\voter)} + 1.  \label{eq:dist2}
	\end{align}

Let $\hat \canL$ and $ \hat \canR$ be two random variables indicating the number of votes that $\canL$ and $\canR$ receive in $\election$, respectively. These two variables are the sum of i.i.d. Bernoulli variables each indicating whether a voter casts a vote or not. Note that all the voters that contribute to $\hat \canR$ are located at the same point, but voters contributing to $\hat \canL$ might have different locations. Using these facts we can calculate the expected value and the variance of $\hat \canL$ and $ \hat \canR$. We have:
	\begin{align*}
	\mathbb{E} [\hat \canR] = \shR &= \frac{q_d}{(2x_d-1)^\beta}, \\
	\mathrm{Var}(\hat \canR) = \sigma^2_\canR &= \frac{q_d}{(2x_d-1)^\beta} \times (1- \frac{1}{(2x_d-1)^\beta}) , \\
	\mathbb{E}[\hat \canL] = \shL &= \sob (1-2x_\voter)^\beta, \\
	\mathrm{Var}(\hat \canL) = \sigma^2_\canL &= \sob (1-2x_\voter)^\beta \times (1-(1-2x_\voter)^\beta).
	\end{align*}

	Let $t = \shL + \frac{\shR}{\sqrt{1+\alpha}}$. Since $t \in [\shL,\shR]$,  we have 
	\begin{equation}
	\mathbb{P}(\hat \canL \geq \hat{\canR}) \leq \mathbb{P}(\hat \canL \geq t) + \mathbb{P}(\hat\canR \leq t). \label{plr}
	\end{equation}
	
	Now, since we know both the expected value and the variance of $\hat \canR$ and $\hat \canL$ we can use Chebyshev's inequality to provide an upper bound on $\mathbb P(\hat \canL \geq \hat\canR)$. 
	
	Chebyshev's inequality states that for a random variable $T$ with finite expected value $\mu$ and finite non-zero variance $\sigma^2$, and for any real number $k > 0$,
\begin{equation}
\mathbb {P}(|T-\mu |\geq k )\leq {\frac {\sigma^2}{k^{2}}}. \label{cheb}
\end{equation}

Therefore we have:
\begin{align}
\mathbb{P}(\hat \canL \geq t) &\leq \mathbb {P}(|\hat \canL-\shL|\geq \frac{\shR}{\sqrt{1+\alpha}}) \nonumber\\
&\leq {\frac {\sigma_\canL^2}{\frac{1}{1+\alpha}\shR^{2}}} \nonumber\\
&\leq {\frac {\sob (1-2x_\voter)^\beta \times (1-(1-2x_\voter)^\beta)}{\shR \times \shL}} \nonumber\\
&\leq {\frac {\sob (1-2x_\voter)^\beta \times (1-(1-2x_\voter)^\beta)}{\shR \times \sob (1-2x_\voter)^\beta}} \nonumber\\
&\leq \frac{1}{\shR}. \label{plgt}
\end{align}
On the other hand,
\begin{align}
\mathbb{P}(\hat \canR \leq t) &\leq \mathbb {P}(|\hat \canR-\shR|\geq \shR-\shL-\frac{\shR}{\sqrt{1+\alpha}})) \nonumber\\
&\leq \mathbb {P}(|\hat \canR-\shR|\geq \shR-\frac{\shR}{1+\alpha}-\frac{\shR}{\sqrt{1+\alpha}})) \nonumber\\
&\leq {\frac {\sigma_\canR^2}{\frac{\alpha^2+1+\alpha - 2\alpha\sqrt{1+\alpha}}{(1+\alpha)^2}\shR^{2}}} \nonumber\\
&= {\frac {\frac{q_d}{(2x_d-1)^\beta} \times (1- \frac{1}{(2x_d-1)^\beta})}{\frac{(\alpha - \sqrt{\alpha + 1})^2}{(1+\alpha)^2}\shR^{2}}} \nonumber\\
&= {\frac {(1- \frac{1}{(2x_d-1)^\beta})}{\frac{(\alpha - \sqrt{\alpha + 1})^2}{(1+\alpha)^2}\shR}}. \label{prlt}
\end{align}
Let 
	$$ \phi(\alpha) = \frac{(\alpha - \sqrt{\alpha + 1})^2}{(1+\alpha)^2}.$$
	Putting Equations \eqref{plr}, \eqref{plgt} and  \eqref{prlt} together we have:
		\begin{equation*}
	\mathbb{P}(\hat \canL \geq \hat{\canR}) \leq \frac{1}{\shR} + \frac{1-\frac{1}{(2x_d-1)^\beta}}{f(\alpha)\shR},
	\end{equation*}
	and by Equation \eqref{eq:dist2} we have:
	\begin{align}
	\dtwo_\beta(\election) &\leq \left(\frac{1}{\shR} + \frac{1-\frac{1}{(2x_d-1)^\beta}}{f(\alpha)\shR}\right) \times \frac{q_d}{\sob(1-2x_\voter)} + 1\nonumber\\
	&= \left(\left(2x_d-1\right)^\beta + \frac{(2x_d-1)^\beta-1}{f(\alpha)}\right) \times \frac{1}{\shL} + 1. \label{temp}
	\end{align}
	Note that since $x_d$ is the only location more distant to $\canL$ than $\canR$, even for $q_d = 1$ the distortion of candidate $\canL$ and consequently the distortion of the election is upper-bounded by $x_d/x_d-1$. Therefore, if $x_d \geq {1 \over 2\alpha} + 1 $, the distortion of the election is upper bounded by $1+2\alpha$ (i.e. $\dtwo_\beta(\election)\leq (1+2\alpha) \done^*_\beta$). So here we assume $x_d < {1 \over 2\alpha} + 1 $. If we substitute ${1 \over 2\alpha} + 1$ for $x_d$ in \eqref{temp} we have:
	
	\begin{align}
	\dtwo_\beta(\election) &\leq \left(\left({1 \over \alpha}+1\right)^\beta + \frac{({1 \over \alpha}+1)^\beta-1}{f(\alpha)}\right) \times \frac{1}{\shL} + 1 \nonumber \\
		&\leq \left({1 \over \alpha}+1\right) \times \left( 1 +\frac{1}{f(\alpha)}\right) \times \frac{1}{\shL} + 1 \nonumber \\
		&= {1 + \alpha \over \alpha} \times \left( 1 +\frac{(1+\alpha)^2}{(\alpha - \sqrt{\alpha + 1})^2}\right) \times \frac{1}{\shL} + 1 \nonumber \\
		&\leq \frac{2 (\alpha + 1)^3}{\alpha (\alpha - \sqrt{\alpha + 1})^2}\times \frac{1}{\shL}+1, \label{finald2}
	\end{align}
	where the last line is due to the fact that
	$$1 \leq \frac{(1+\alpha)^2}{(\alpha - \sqrt{\alpha + 1})^2}.$$
	
	Now, suppose that the the expected number of votes that each candidate receives is large enough, so that 
	\begin{equation*}
		\shL \geq \frac{(\alpha + 1)^3}{\alpha^2 (\alpha - \sqrt{\alpha + 1})^2}.
	\end{equation*}
	By Equation \eqref{finald2} we have:
		\begin{align*}
		\dtwo_\beta(\election) &\leq \frac{2 (\alpha + 1)^3}{\alpha (\alpha - \sqrt{\alpha + 1})^2}\times \frac{\alpha^2 (\alpha - \sqrt{\alpha + 1})^2}{(\alpha + 1)^3}+1 \\
		&\leq 1+2\alpha \\ &\leq (1+2\alpha) \done^*_\beta.
		\end{align*}
	This completes the proof.
\end{proof}

Now, we are ready to prove Theorem \ref{bbnd}. 

\begin{proof}[Proof of Theorem \ref{bbnd}.]
	Fix any $\alpha>0$ and $\beta\in [0,1]$, and let $\election \in \Omega_\beta$ be an arbitrary election whose candidates receive at least
	$$\frac{(\alpha + 1)^3}{\alpha^2 (\alpha - \sqrt{\alpha + 1})^2}$$ 
	expected number of votes. 
	Recall that our assumption is that $\canR$ is both the optimal and the expected winner. Now, based on the value of $\shR / \shL$, there are two cases: either $\shR / \shL \leq 1+\alpha$ or $\shR / \shL > 1+\alpha$. For the first case, by Lemma \ref{lbnd}  value of $\dtwo_\beta(\election)$ is upper bounded by $(1+2\alpha) \done^*_\beta$. For the second case, since  
	$$
	\shL \geq \frac{(\alpha + 1)^3}{\alpha^2 (\alpha - \sqrt{\alpha + 1})^2},
	$$
	by Lemma \ref{fuch},  
	the expected distortion is upper bounded by $(1+2\alpha) \done^*_\beta$. 
	Combining these two cases yields the upper-bound of  $(1+2\alpha)\done^*_\beta$ on $\dtwo_\beta(\election)$.
\end{proof}

As an example, for $\alpha = 0.1$, Theorem \ref{bbnd} states that for every election $\election$ whose candidates receive at least $148$ expected number of votes, the expected distortion is upper bounded by $1.2\done^*_\beta$ which for $\beta=1$ is $ \simeq 1.83$.

We complement Theorem \ref{bbnd} by describing how to construct bad examples with expected distortion value near $\done^*_\beta$.


\begin{example}
	\label{exmpl}
	Consider Optimization Problem \ref{cp2}, with an additional constraint that $\shL \geq \shR (1+\varepsilon)$ for a fixed constant $\varepsilon$, and let $\done^{**}_\beta$ be the answer of this optimization problem and $\election^{**}$ be its corresponding election. By Chernoff bound, for a large enough value of $\shL$, candidate $\canL$ wins the election with a high probability, i.e., 
	$$\lim_{\shL \rightarrow \infty} \dtwo_\beta(\election^{**}) \simeq D(\canL) \simeq \done^*_\beta.$$  
\end{example}

Note that, the bound provided by Theorem \ref{bbnd} is almost tight; as the election size grows, the upper bounds of Theorem \ref{bbnd} tends to the distortion value of  Example \ref{exmpl}. However, for elections with a small number of voters, the distortion value might be larger. For example, consider a simple scenario where there is one voter located at point $1+\varepsilon \in \regD$ and $\beta = 1$ (see Figure \ref{fig:bad}). For this case, the distortion value is
\begin{align*}
\probL\cdot \frac{\cost_{\election}(\canL)}{\cost_{\election}(\canR)} + \probR &= \probL \cdot \frac{1+\varepsilon}{\varepsilon} + \probR \\
&= \frac{\varepsilon}{1+2\varepsilon} \cdot \frac{1+\varepsilon}{\varepsilon} + \frac{1+\varepsilon}{1+2\varepsilon}\\
&= \frac{2+2\varepsilon}{1+2\varepsilon},
\end{align*}
which tends to $2$ as $\varepsilon \rightarrow 0$. We conjecture that this example is the worst possible scenario and value of $\dtwo^*_\beta$ is upper bounded by $2$ for any election with any size while $\beta = 1$. 

\begin{figure}[h]
	\begin{center}
		\includegraphics[scale=0.9]{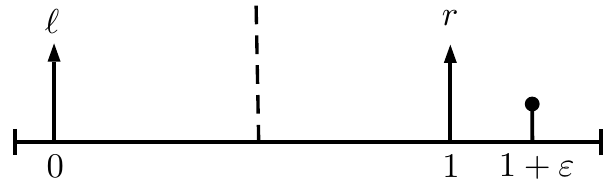}
	\end{center}
	\caption{An example with maximum expected distortion. $\dtwo_\beta(\election)$ for $\beta=1$  tends to $2$ as $\varepsilon\rightarrow 0$.}
	\label{fig:bad}
\end{figure}

\subsection{Valid Displacements}\label{validdisp2}

In this section, we prove Lemmas \ref{vd1} and \ref{toward}. 
\vdfirst*
\begin{proof}
	We prove the statement of Lemma \ref{vd1} for regions $\regA$ and $\regB$. Similar arguments can be used to prove the lemma for regions $\regC$ and $\regD$.
	Let $x_i$ be the current location of $\voter_i$ in region $\regA$ ($x_i<0$). By definition, $\voter_i$ casts a vote with probability $1 / \left(1-2x_i\right)^\beta.$ 
	Now, consider point $x = -x_i/\left(1-2x_i\right)$. We claim that an agent at $x$, votes for $\canL$ with the same probability as $\voter_i$. First, note that since $x_i<0$, $$0\leq-x_i/\left(1-2x_i\right)\leq 1/2.$$
	Hence, the preferred candidate of the voter located at $x$ is $\canL$. Furthermore, for any agent at $x$, the probability of casting a vote is
	\begin{align*}
	(1-2x)^\beta &= \left(1- \frac{-2x_i}{1-2x_i}\right)^\beta \\
	&=\left(\frac{1-2x_i+2x_i}{1-2x_i} \right)^\beta \\
	&=\left(\frac{1}{1-2x_i}\right)^\beta .
	\end{align*}

	Therefore, by moving $v_i$ from $x_i \in \regA$ to $x'_i= \frac{-x_i   }{1-2x_i} $, the probability that $v_i$ votes for $\canL$ remains the same. Let $\election'$ be the election, after moving $v_i$ to $x$.
	We have
	\begin{align*}
	\frac{\cost_{\election'}(\ell)}{\cost_{\election'}(r)} = \frac{\cost_{\election}\left(\canL\right) - \big[-x_i - \left(-x_i/\left(1-2x_i\right)\right)\big]}{\cost_{\election}\left(\canR\right)  -\big[\left(1-x_i\right) - \left(-x_i/\left(1-2x_i\right)\right)\big]}.	
	\end{align*}
	Since we have
	$$\frac{-x_i - \left(-x_i/\left(1-2x_i\right)\right)}{\left(1-x_i\right) - \left(-x_i/\left(1-2x_i\right)\right)} \leq 1 \leq \frac{\cost\left(\canL\right)}{\cost\left(\canR\right)},$$
	using Observation \ref{lem:frac} we conclude that 
	$
		\frac{\cost_{\election'}(\ell)}{\cost_{\election'}(r)} \geq \frac{\cost_{\election}(\ell)}{\cost_{\election}(r)}
			$ which in turn implies that
	moving $\voter_i$ to $x$ is a valid displacement.
\end{proof}

\twd*
\begin{proof}
	Assume without loss of generality that $x_i < x_j$. We show that we can move both the voters to point $$t=\frac{\sqrt{\left(2x_i-1\right)\left(2x_j-1\right)} +1 }{2}.$$
	
	Let $\hat \voter_k$ be a random variable which is equal to $1$, if $\voter_k$ casts a vote and $0$ otherwise. In addition, let $$A_k = \mathbb{P} \left(\canR \text{ wins the election}| \hat \voter_i + \hat\voter_j  = k\right)$$ where $0 \leq k \leq 2$. Trivially, we have $A_0 \leq A_1 \leq A_2$, and 
	\begin{align}
	\probR &= A_0 \cdot \mathbb{P} \left(\hat\voter_i + \hat\voter_j  = 0\right) \nonumber\\&+ A_1 \cdot \mathbb P\left(\hat\voter_i + \hat\voter_j  = 1\right) \nonumber\\&+ A_2 \cdot \mathbb P\left(\hat\voter_i + \hat\voter_j  = 2\right). \label{prr}
	\end{align}
	Furthermore, note that we have 
\begin{align*}
	\mathbb P\left(\hat \voter_i+\hat \voter_j=0\right) &= \left(1-  \frac{1}{(2x_i-1)^\beta} \right) \left(1-  \frac{1}{(2x_j-1)^\beta} \right)\\
	&= \frac{1+\left(2x_i-1\right)^\beta\left(2x_j-1\right)^\beta-\left(2x_i-1\right)^\beta-\left(2x_j-1\right)^\beta} {\left(2x_i-1\right)^\beta\left(2x_j-1\right)^\beta},\\
	\mathbb P\left(\hat \voter_i+\hat \voter_j=2\right) &= \frac{1}{\left(2x_i-1\right)^\beta\left(2x_j-1\right)^\beta}.
\end{align*}

	Let $\hat \voter'_i$ and $\hat \voter'_j$ be variables indicating whether $\voter_i$ and $\voter_j$ cast a vote or not, after the displacement. We have 
	\begin{align*}
	\mathbb P\left(\hat \voter'_i+\hat \voter'_j=0\right) &= \left(1-  \frac{1}{(2t-1)^\beta}\right)^2\\
	&= \left(1-  \frac{1}{(\sqrt{\left(2x_i-1\right)\left(2x_j-1\right)})^\beta}\right)^2\\
	&= \frac{1 + \left(2x_i-1\right)^\beta\left(2x_j-1\right)^\beta - 2\sqrt{\left(2x_i-1\right)^\beta\left(2x_j-1\right)^\beta} }{\left(2x_i-1\right)^\beta\left(2x_j-1\right)^\beta},\\
	\mathbb P\left(\hat \voter'_i+\hat \voter'_j=2\right) &= \frac{1}{\left(2t-1\right)^{2\beta}} = \frac{1}{\left(2x_i-1\right)^\beta\left(2x_j-1\right)^\beta}.
	\end{align*}

	Thus, we have $$\mathbb P\left(\hat \voter_i+\hat \voter_j=2\right) = \mathbb P\left(\hat \voter'_i+\hat \voter'_j=2\right).$$
%
Now, we show 
	$$\mathbb P(\hat \voter_i+\hat \voter_j=0)  \leq \mathbb P(\hat \voter'_i+\hat \voter'_j=0).$$
	We have
	\begin{align*}
	\mathbb P(\hat \voter_i+\hat \voter_j=0) - \mathbb P(\hat \voter'_i+\hat \voter'_j=0) &= \frac{2\sqrt{\left(2x_i-1\right)^\beta\left(2x_j-1\right)^\beta}-\left(2x_i-1\right)^\beta-\left(2x_j-1\right)^\beta}{\left(2x_i-1\right)^\beta\left(2x_j-1\right)^\beta}.
	\end{align*}
	Since $(2x_i-1)^\beta(2x_j-1)^\beta > 0$ we just need to show 
	$$
2\sqrt{\left(2x_i-1\right)^\beta\left(2x_j-1\right)^\beta}-\left(2x_i-1\right)^\beta-\left(2x_j-1\right)^\beta \leq 0,
	$$
	which is trivial due to the fact that
	$$
	2\sqrt{\left(2x_i-1\right)^\beta\left(2x_j-1\right)^\beta}-\left(2x_i-1\right)^\beta-\left(2x_j-1\right)^\beta = - \left(\sqrt{\left(2x_i-1\right)^\beta} - \sqrt{\left(2x_j-1\right)^\beta}\right)^2.
	$$
	Furthermore, since $$  \sum_{0 \leq k \leq 2} \mathbb P\left(\hat \voter'_i+\hat \voter'_j=k\right) = 1,$$ we have 
	$$
	\mathbb P\left(\hat \voter_i+\hat \voter_j=1\right) > \mathbb P\left(\hat \voter'_i+\hat \voter'_j=1\right).
	$$
	
	Considering Equation \eqref{prr}, and the fact that $A_0 \leq A_1$ we  conclude that after this movement, the value of $\probR$ decreases and the value of $\probL$ increases.
	
	Finally, let $C$ and $C'$ be the cost of the agents other than $\voter_i$ and $\voter_j$ for $\canL$ and $\canR$, respectively. By definition, before the displacement, we have 
	$$
	D\left(\canL\right) = \frac{C + [x_i +x_j]}{C' + [x_i +x_j-2]}
	$$
	and  after moving $\voter_i$ and $\voter_j$ to point $t$, we have:
	$$
	D \left(\canL\right) = \frac{C + [\sqrt{\left(2x_i-1\right)\left(2x_j-1\right)}+1]}{C' + [\sqrt{\left(2x_i-1\right)\left(2x_j-1\right)}-1]}.
	$$
	Again, by straightforward calculus, one can  easily verify that
	
	$$
	x_i + x_j \geq \sqrt{\left(2x_i-1\right)\left(2x_j-1\right)}+1,
	$$
	
	Thus, after this displacement, both $D\left(\canL\right)$ and $\probL$ increases, and so does the value of $ \dtwo_\beta\left(\election\right)$. 
\end{proof}

\section{General Metric}
\label{general}
We now extend our results to general metric spaces. Suppose that the voters and candidates are located in an arbitrary metric $\metric$. By definition, for every voter $\voter_i$ and candidates $\canL,\canR$ we have: 

\begin{itemize}
	\item $d_{i,\canL},d_{i,\canR} \geq 0$.
	\item $d_{i,\canL} + d_{i,\canR} \geq d_{\canL,\canR}$ (triangle inequality).
\end{itemize} 
We suppose without loss of generality that $d_{\canL,\canR} = 1$. For this case, we prove Theorem \ref{thm:general}, which states that for every  $0 \leq \beta \leq 1$, the same upper bounds we obtained on the distortion value for the line metric also works for any arbitrary metric space.  

\begin{theorem}
	\label{thm:general}
For every election $\election$ in an arbitrary metric space, there exists an election $\election' $ in line metric, such that $\done(\omega_{\election,\beta}) \leq \done\left(\omega_{\election',\beta}\right)$ and $\dtwo_\beta\left(\election\right) \leq \dtwo_\beta\left(\election'\right)$. 	
\end{theorem}

\begin{proof}
	Let $\election$ be an  election in an arbitrary metric space $\metric_\election$. Assume w.l.o.g. that candidate $\canR$ is the optimal candidate and let $\voters_\election$ be the set of voters in election $\election$. For each voter $\voter_i \in \voters_\election$, let  $\gamma_i = d_{i,\canL} / d_{i,\canR}$. Based on the value of $c_i$, we partition the voters into two subsets $\voters^+_\election$ and $\voters^-_\election$, where 
	$$
	\voters^-_\election = \{\voter_i | \gamma_i \leq D\left(\canL\right) \} $$ $$
	\voters^+_\election = \{\voter_i | \gamma_i > D\left(\canL\right) \}.
	$$    
	Now, we construct election $\election'$ as follows: consider a line and two candidates $\canL',\canR'$ located respectively at $0$ and $1$.     
	For each voter $\voter_i \in \voters^-$, we consider a voter $\voter'_i$ in $\election'$, located at point 
	$
	x'_i = \frac{c_i}{c_i+1}.
	$
	Since 
	\begin{align*}
	\left(\frac{|d_{i,\canL'}-d_{i,\canR'}|}{d_{i,\canL'} + d_{i,\canR'}}\right)^\beta &= \left(\frac{|2x'_i-1|}{1}\right)^\beta\\
	&=  \bigg|2\frac{\gamma_i}{\gamma_i+1}-1\bigg|^\beta \\
	&= \bigg|2 \frac{d_{i,\canL} / d_{i,\canR}}{d_{i,\canL} / d_{i,\canR}+1}-1\bigg|^\beta \\
	&= \bigg| \frac{d_{i,\canL}-d_{i, R}}{d_{i,\canL} + d_{i, R}}\bigg|^\beta,
	\end{align*}
	both $\voter_i$ and $\voter'_i$ participate in their corresponding elections with equal probabilities.
	Similarly, for each voter $\voter_i \in \voters^+$, we consider a voter $\voter'_i$ located at point 
	$
	x_i = \frac{\gamma_i}{\gamma_i-1}.
	$
	Again, it can be observed that 
	$$
	\left(\frac{|d_{i,\canL'}-d_{i,\canR'}|}{d_{i,\canL'} + d_{i,\canR'}}\right)^\beta = \left(\frac{|d_{i,\canL}-d_{i,\canR}|}{d_{i,\canL} + d_{i,\canR}}\right)^\beta.
	$$
	In conclusion, for every $i$, voters $\voter_i$ and $\voter'_i$ cast a vote in their corresponding elections with equal probabilities. Thus,  expected winners of $\election'$ and $\election$ are the same, and we have 
	\begin{align}
	\mathtt{P}_{\canL'} &= \probL, \quad 
	\mathtt{P}_{\canR'} = \probR. \label{fu}
	\end{align}  
	Now, we prove $D\left(\canL\right) \leq D\left(\canL'\right)$. For convenience, let 
	$$
	A= \sum_{\voter_i \in \voters^-} d_{i,\canL} \qquad A'= \sum_{\voter_i \in \voters^-} d_{i,\canL'}
	$$ 
	$$
	B= \sum_{\voter_i \in \voters^-} d_{i,\canR} \qquad B'= \sum_{\voter_i \in \voters^-} d_{i,\canR'}
	$$ 
	$$
	C= \sum_{\voter_i \in \voters^+} d_{i,\canL} \qquad C'= \sum_{\voter_i \in \voters^+} d_{i,\canL'}
	$$ 
	$$
	D= \sum_{\voter_i \in \voters^+} d_{i,\canR} \qquad D'= \sum_{\voter_i \in \voters^+} d_{i,\canR'}.
	$$
	Note that for each $\voter_i \in \voters^-$, $d_{i,\canL'} = d_{i,\canL}/(d_{i,\canR}+d_{i,\canL})$,
	$d_{i,\canR'} = d_{i,\canR}/(d_{i,\canR}+d_{i,\canL})$, and $d_{i,\canR}+d_{i,\canL} \geq 1$. Hence $d_{i,\canL} \geq d_{i,\canL'}$ and $d_{i,\canR} \geq d_{i,\canR'}$. Therefore we have $A-A' \geq 0$ and $B-B' \geq 0$. In addition, 
	\begin{align}
	\frac{A-A'}{B-B'} &= \frac{\sum_{i \in \voters^-}d_{i,\canL} - \frac{\gamma_i}{\gamma_i+1}}{\sum_{i \in \voters^-}d_{i,\canR}- \frac{1}{\gamma_i+1}} \nonumber \\
	&= \frac{\sum_{i \in \voters^-} \frac{d_{i,\canL} + \gamma_i d_{i,\canL}- \gamma_i}{\gamma_i+1}}{\sum_{i \in \voters^-} \frac{ d_{i,\canR}+ \gamma_i d_{i,\canR}-1}{\gamma_i+1}} \nonumber\\
	&\leq \max_{i \in \voters^-} \frac{d_{i,\canL} + \gamma_i d_{i,\canL}- \gamma_i}{ d_{i,\canR}+ \gamma_i d_{i,\canR}-1} \nonumber\\
	&= \max_{i \in \voters^-} \frac{\gamma_i d_{i,\canR} + \gamma_i d_{i,\canL}- \gamma_i}{ d_{i,\canR}+ d_{i,\canL}-1} &(\gamma_i d_{i,\canR} = d_{i,\canL}) \nonumber\\
	&\leq \max_{i \in \voters^-} \gamma_i \nonumber\\
	&\leq D\left(\canL\right) \label{gm_first}
	\end{align}
	On the other hand, for each $\voter_i \in \voters^+$, $d_{i,\canL'} = d_{i,\canL}/(d_{i,\canL}-d_{i,\canR})$, $d_{i,\canR'} = d_{i,\canR}/(d_{i,\canL}-d_{i,\canR})$, and $d_{i,\canL}-d_{i,\canR} \leq 1$. Hence $d_{i,\canL} \leq d_{i,\canL'}$ and $d_{i,\canR} \leq d_{i,\canR'}$. Therefore we have $C'-C \geq 0$ and $D'-D \geq 0$. Furthermore, we have:
			\begin{align}
		\frac{C'-C}{D'-D} &= \frac{\sum_{i \in \voters^+} \frac{\gamma_i}{\gamma_i-1}-d_{i,\canL}}{\sum_{i \in \voters^+}\frac{1}{\gamma_i-1}-d_{i,\canR}} \nonumber \\
		&= \frac{\sum_{i \in \voters^+} \frac{\gamma_i-\gamma_i d_{i, \canL}+d_{i, \canL}}{\gamma_i-1}}{\sum_{i \in \voters^+}\frac{1-\gamma_id_{i,\canR}+d_{i,\canR}}{\gamma_i-1}} \nonumber \\
		&\geq \min_{i \in \voters^+} \frac{\gamma_i-\gamma_i d_{i, \canL}+d_{i, \canL}}{1-\gamma_id_{i,\canR}+d_{i,\canR}} \nonumber \\
		&= \min_{i \in \voters^+} \frac{\gamma_i-\gamma_i d_{i, \canL}+\gamma_id_{i, \canR}}{1-d_{i,\canL}+d_{i,\canR}} &(\gamma_i d_{i,\canR} = d_{i,\canL})  \nonumber \\
		&\geq \min_{i \in \voters^+} \gamma_i \nonumber \\
		&\geq D(\canL). \label{gm_second}
		\end{align}

	By Equations \eqref{gm_first} and \eqref{gm_second}, and using Observation \ref{lem:frac} we have 
\begin{equation}
\frac{(C'-C)-(A-A')}{(D'-D)-(B-B')} \geq D\left(\canL\right), \label{gm_third}
\end{equation}
and
\begin{align}
D\left(\canL\right) &= \frac{A+C}{B+D} \nonumber \\
&\leq \frac{(A+C)+(C'-C)-(A-A')}{(B+D)+(D'-D)-(B-B')} & \text{(Observation \ref{lem:frac} and Equation \eqref{gm_third})} \nonumber\\
&= \frac{C' + A'}{D'+B'} \nonumber\\
&= D(\canL'). 	\label{fc}
\end{align}
	
	Since the expected winner is the same in $\election$ and $\election'$, Inequality  \eqref{fc} immediately implies that $ \done\left(\omega_{\election,\beta}\right) \leq \done\left(\omega_{\election',\beta}\right).$ Furthermore, considering Equations \eqref{eddef} ,\eqref{fu}, and \eqref{fc} we have $\dtwo_\beta\left(\election\right) \leq \dtwo_\beta\left(\election'\right).$ 
\end{proof}

%

\section{Future Directions}
In this study, we analyzed the distortion value in a spatial voting model with two candidates, when the voters are allowed to abstain. The set of results in this paper provides a rather complete picture of the model. Nevertheless, some important open questions remain open.

\begin{itemize}
	\item The most immediate open question is to analyze the expected distortion value of the elections for a small number of voters. The counter-example in Section \ref{aatb} refutes the existence of an upper bound better than $2$. We believe that this example is the worst possible scenario. However, we don't have a formal proof for this claim.

	\item Another direction is to provide a closed-form expression for the distortion of the expected winner. Currently, the maximum distortion is obtained via a mathematical program, which is not even convex \footnote{Of course, we have a short note on how to reduce this program into a convex one, by eliminating some of the variables.}.

\end{itemize}

Beyond the above direct questions,  this research also initiates an interesting line of work and opens a fruitful  direction for the future research. In the following, we discuss two of these directions:

\begin{itemize}
	\item  In this paper, we focused on a majority election between two candidates. When  more than two  candidates are running, vote aggregation becomes more complex. One interesting direction is to generalize the models in this paper for elections with more than two candidates and analyze the performance of different well-established voting mechanisms such as Borda, $k$-approval, Veto, Ranked pairs,  and Copland under abstention assumption. One can also consider abstention in evaluating the distortion of different randomized mechanisms. 
	
	\item 
	Similar to the elections with no abstention, it seems that high distortion scenarios stem from the issue of representativeness of candidates. Cheng et. al. \cite{cheng2017distortion} show that when the candidates are of the people (i.e., they have the same distribution as the voters), distortion ratio improves to a constant upper-bound strictly better than $2$ for general metrics.  The question is, how does the distortion value change if we allow abstention in the societies that voters and candidates have the same distribution? 
\end{itemize}

\vskip 0.2in

\bibliographystyle{plain}
\bibliography{voting}

\end{document}